\documentclass[letterpaper,12pt]{article}

\usepackage{microtype}

\usepackage{fullpage}
\usepackage{color}
\usepackage{latexsym,amssymb,amsthm,amsmath}

\usepackage{graphicx}

\usepackage[cp850]{inputenc}

\usepackage{enumerate} % smarter enumerate environment

\usepackage{cite} % Citations stick to previous text, are sorted and
\newtheorem{theorem}{Theorem}
\newtheorem{proposition}{Proposition}
\newtheorem{lemma}{Lemma}
\newtheorem{corollary}[theorem]{Corollary}

\newtheorem{observation}{Observation}

\newcommand{\reals}{\mathbb{R}}
\renewcommand{\P}{\mathcal{P}}

\newcommand{\D}{\mathcal{D}}
\renewcommand{\S}{\mathcal{S}}

\DeclareMathOperator{\Vor}{Vor} % I prefer this over VD.
\DeclareMathOperator{\DT}{DT} % I do not care if it's DT or Del or
\DeclareMathOperator{\DG}{DG} % Same here.
\DeclareMathOperator{\SG}{SG} % Same here.
\DeclareMathOperator{\MNG}{MNG} % Same here.
\DeclareMathOperator{\CH}{CH} % Same here.
\DeclareMathOperator{\RIG}{RIG} % Same here.

\let\setminus\smallsetminus
\let\emptyset\varnothing
\let\bd\partial
\title{Witness (Delaunay) Graphs}

\author{Boris Aronov%
  \thanks{Department of Computer Science and Engineering, Polytechnic
    Institute of NYU, Brooklyn, New York~~11201, USA.  Research
    partially supported by a grant from the U.S.-Israel Binational
    Science Foundation, by NSA MSP Grant H98230-06-1-0016, and NSF
    Grant CCF-08-30691.}
  \and
  Muriel Dulieu\footnotemark[1]
  \and
  Ferran Hurtado\thanks{Departament de Matem\`{a}tica Aplicada II, Universitat Polit\`{e}cnica de Catalunya,
    Barcelona, Spain.
    Partially supported by projects MEC MTM2006-01267,
    MTM2009-07242, Gen.~Catalunya DGR 2005SGR00692 and 2009SGR1040.}}

\date{17 May 2010}

\begin{document}
%-----------------------------------------------------------------------------

\maketitle
\begin{abstract}
Proximity graphs are used in several areas in which a
neighborliness relationship for input data sets is a useful tool
in their analysis, and have also
received substantial attention from the graph drawing community,
as they are a natural way of implicitly representing graphs.
However, as a tool for graph representation, proximity graphs have
some limitations that may be overcome with suitable
generalizations.

We introduce a generalization, \emph{witness graphs}, that
encompasses both the goal of more power and flexibility for graph
drawing issues and a wider spectrum for neighborhood analysis. We
study in detail two concrete examples, both related to Delaunay
graphs, and consider as well some problems on stabbing geometric
objects and point set discrimination, that can be naturally described in
terms of witness graphs.
\end{abstract}

\thispagestyle{plain}

%%%%%%%%%%%%%%%%%%%%%%%%%%%%%%%%%%%%%%%%%%%%%%%%%%%%%%%%%%%%%%%%%%%%%%%%%%%%%%
\section{Introduction and preliminary definitions}\label{section:introduction}

\emph{Proximity graphs} are used in several areas in which a
neighborliness relationship for input data sets is a useful tool
in their analysis and use, see \cite{JT92} for a survey.  Examples
of such areas are computer vision, geographic analysis, pattern
classification, computational morphology, and spatial analysis. On
the other hand, proximity graphs have also received substantial
attention from the graph drawing community, as they are a natural
way of implicitly representing graphs; a survey
of such results appeared in \cite{BLL94} and has been extended and
updated in \cite{Li08}.

As a tool for graph representation, proximity graphs have some
limitations that may be overcome with suitable generalizations. An
example of such an extension is the concept of \emph{weak
proximity graphs} \cite{BLW06}. Here we introduce a generalization
that encompasses both the goal of more power and flexibility for
graph drawing issues and a wider spectrum for neighborhood
analysis.

In general, given a point set $P$ and a set of geometric shapes
$S$, a proximity graph is a graph $G=(P,E)$ with $P$ as the vertex
set and two points $a$ and $b$ being adjacent if and only if there is a suitable shape
$\Gamma$, defined by $a$ and $b$, from $S$---their \emph{region of influence}---that covers
them but no other point from $P$; the presence of another point is
referred to as an \emph{interference}. Whether the two points have
to be on the boundary of the shape $\Gamma$, whether $\Gamma$ is
uniquely determined by them, and whether the interference is considered
only if interior to $\Gamma$, depends on the specific problem
studied as also does the family of shapes under consideration; see
\cite{JT92,Li08} for an extensive list of examples of proximity
graphs.

A \emph{witness graph} $G=(V,E)$ is defined by a quadruple
$(P,S,W,\pm)$ in which $P=V$ is the set of \emph{vertex points} (or
just \emph{vertices}), $S$ provides the geometric shapes, and $W$ is a
second point set, consisting of the \emph{witness points} (or just
\emph{witnesses}). In the \emph{positive witness} version~($+$), the
tentative adjacency between $a$ and $b$ is accepted if and only if a
witness point is covered by at least one of the regions of influence
defined by $a$ and $b$. In the \emph{negative witness} version ($-$),
a witness inside the interaction region would destroy the tentative
adjacency, hence there is an adjacency between $a$ and $b$ if at least
one of their regions of influence is free of any witness. Notice that
in both cases we only pay attention to the presence of witnesses in
the regions of influence, not of points from $P$. In a third variation
one may admit the presence of both negative and positive witnesses and
use a combined decision rule; we do not pursue this possibility here.

To the best of our knowledge this family of graphs has not been
introduced before in its full generality, yet, not surprisingly, the
situation has been considered in more or less explicit form for some
specific graphs.  Ichino and Slansky \cite{IS85} defined the
\emph{rectangular influence graph}, $\RIG(P)$, in which two points
$p,q\in P$ are adjacent when the rectangle having them as opposite
corners (the \emph{box} they define) contains no point from $P$. In
the same paper, they defined the \emph{mutual neighborhood graph}
$\MNG(P|Q)$, in which $p,q\in P$ are adjacent when the associated box
contains no point from $Q$, and they studied some properties that can
be derived by considering simultaneously $\MNG(P|Q)$ and $\MNG(Q|P)$. In
\cite{BCO92}, De~Berg, Cheong and Overmars defined the \emph{dominance
  in a set $P$ with respect to a set $Q$} and gave an efficient
algorithm for its computation: $a\in P$ dominates $b\in P$ when
$x(a)\ge x(b)$, $y(a)\ge y(b)$, and the box defined by $a$ and $b$
contains no point from $Q$.  Finally, McMorris and Wang \cite{MW00}
defined the \emph{sphere-of-attraction graphs} in which from every
point of $p\in P$ taken as center a ball is grown until a first point
from $Q$ is encountered; the graph is then defined on $P$ as a ball
intersection graph. They obtained a characterization in dimension one
and initiated the study in higher dimensions.

In the present paper, we consider two concrete examples, both
related to Delaunay graphs, one for positive witnesses and one for
negative ones.  Other witness graphs such as the \emph{witness
Gabriel graph} and the \emph{witness
  rectangle-of-influence graph} are studied in the companion papers
\cite{ADH09,ADH08}. A systematic study is developed in \cite{thesis}.

We define the \emph{witness Delaunay graph} of a point set $P$ of
\emph{vertices} in the plane, with respect to a point set $W$ of
\emph{witnesses}, denoted $\DG^{-}(P,W)$, as the graph with vertex
set $P$ in which two points $x,y\in P$ are adjacent if and only if
there is an open disk that does not contain any witness $w\in W$
whose bounding circle passes through $x$ and $y$. It is a
negative-witness graph in which  the shapes are all the disks in
the plane whose boundary contains two points from $P$.  When
$W=\varnothing$ the graph $\DG^{-}(P,\varnothing)$ is simply the
complete graph $K_{|P|}$. When $W=P$ the graph $\DG^{-}(P,P)$ is
precisely the Delaunay graph $\DG(P)$, which under standard
non-degeneracy assumptions is a triangulation and is denoted
$\DT(P)$ (see, e.g., \cite{AK00,Fo04}).  The latter example
illustrates the fact that the use of a witness set gives a
generalization of the basic Delaunay structure. The properties of
$\DG^{-}(P,W)$ are studied in Section~\ref{section:delaunay}.

The \emph{square graph} of a point set $P$ in the plane, with
respect to a point set $W$ of witnesses, denoted $\SG^{+}(P,W)$, is
the graph with vertex set~$P$, in which two points $x,y\in P$ are
adjacent when there is an axis-aligned square with $x$ and $y$ on
its boundary whose interior contains some witness point $q\in W$.
It is a positive-witness graph in which the shapes are all the
 axis-aligned squares in the plane whose boundary contains two points from $P$.
Observe that a negative-witness version $\SG^-(P,W)$ of this
graph, with $W=P$, would be the standard Delaunay graph for the
$L_{\infty}$ metric, and hence we are studying here the
positive-witness--variant of this Delaunay structure. The graph
$\SG^{+}(P,W)$ is discussed in Section~\ref{section:squares}.

In this work we describe algorithms for the computation of these
graphs and prove several of their fundamental properties.  We also
give a complete characterization of the combinatorial graphs that
admit a realization as $\SG^{+}(P,W)$ for suitable sets $P$ and $W$, a
kind of result that, however, remains elusive for $\DG^{-}(P,W)$.
In Section~\ref{section:stabbing}, we also present some
related results on stabbing geometric objects, which can
be essentially described as follows: given a point set $P$,
find a second point set $W$, as small as possible, such that
no pair of points $p,q\in P$ have adjacent regions in the
Voronoi diagram of $P\cup W$.

We use standard graph terminology as in \cite{CL04}; in
particular, for a graph $G=(V,E)$ we write $xy\in E$ or $x\sim y$
to indicate that $x,y\in V$ are adjacent vertices of $G$.
The terms \emph{closed} and \emph{open} are used in the
sense of closed and open sets (sets with or without their boundary).

%%%%%%%%%%%%%%%%%%%%%%%%%%%%%%%%%%%%%%%%%%%%%%%%%%%%%%%%%%%%%%%%%%%%%%%%%%%%%%
\section{Witness Delaunay graphs}\label{section:delaunay}

Consider a witness Delaunay graph $\DG^{-}(P,W)$ of a point set $P$
with respect to a witness set $W$.  We assume that the set $P \cup W$
is \emph{in general position}, i.e., that no three distinct
points in $P\cup W$ are collinear and that no four distinct points in
$P\cup W$ are concyclic.  We denote by $E$ the edge set of the graph,
that will be drawn as segments as usual for Delaunay graphs.  Let $n
\mathop{:=} \max\{|P|,|W|\}$.
We say that a disk \emph{covers} a witness if the witness lies in its interior.

Note that, by definition of the witness Delaunay graph, the presence of an edge
between vertices $p,q \in P$ is independent of the fact that $p$
and/or $q$ might be witnesses, since any open disk whose boundary
passes through $p$ and $q$ does not cover either point.

First, a simple geometric observation:
\begin{observation} \label{obs:shrunken-disk}
  If $D$ is a closed disk containing points $p$ and $q$ then there exists a
  disk $D_{pq}\subset D$ whose boundary passes through $p$ and $q$.
\end{observation}
\begin{proof}
  Let $c$ be the center of $D$.  Shrink~$D$ while keeping its center
  at $c$ until it is about to lose $p$ or $q$.  Let the resulting disk
  be $D'$.  Without loss of generality, let $p \in \bd D'$.
  Shrink~$D'$ by a homothety with center $p$ until it is about to lose
  $q$.  The result is the desired disk~$D_{pq}$.
\end{proof}

We start with the computation of the witness Delaunay graph, which
requires some lemmas; the first one is immediate from the
definition of $\DG^{-}(P,W)$:

\begin{lemma}
  \label{lem:adjacency-condition}
  Two points $p, q\in P$ are adjacent in $\DG^{-}(P,W)$ if and only if
  they are neighbors in $\DT(W\cup\{p,q\})$.
\end{lemma}

\begin{lemma} \label{lem:neighborsBySector}
Let $w_1,\dots,w_t$ be the Delaunay neighbors of $p\in P$ in
$\DT(W\cup\{p\})$ given in counterclockwise radial order, and let
$q\in P$ be a point whose radial position around $p$ is between
$w_i$ and $w_{i+1}$.  If $\measuredangle w_ipw_{i+1}\ge\pi$, then
$p$ and  $q$ are adjacent in $\DG^{-}(P,W)$; if $\measuredangle
w_ipw_{i+1}<\pi$, then $p$ and  $q$ are adjacent in $\DG^{-}(P,W)$
if and only if $q$ lies in the interior of the circle through $p$,
$w_i$, and $w_{i+1}$.
\end{lemma}
\begin{proof}
  If $\measuredangle w_ipw_{i+1}\ge\pi$, then $p$ must be a vertex of
  the convex hull $\CH(W\cup\{p\})$ and the segment $pq$ is external
  to this hull. Therefore there is a disk (in fact, a half-plane)
  containing $p$ and $q$ but covering no point from $W$, so they are
  adjacent in $\DG^{-}(P,W)$ by Observation~\ref{obs:shrunken-disk}.
  Assume now that $\measuredangle w_ipw_{i+1}<\pi$; then $pw_iw_{i+1}$
  is a triangle in $\DT(W\cup\{p\})$ whose circumscribing disk $D$
  covers no points from $W$.  If $q$ is exterior to $D$ then $w_i$
  and $w_{i+1}$ are neighbors in $\DT(W\cup\{p,q\})$ and $p$ and $q$
  cannot be adjacent in $\DG^{-}(P,W)$ because the segments $pq$ and
  $w_iw_{i+1}$ cross. If $q$ is interior to $D$ then
  $p$ and $q$ are adjacent in $\DG^{-}(P,W)$, by
  Observation~\ref{obs:shrunken-disk}.
 % $q$ is adjacent
 %  to $p$, $w_i$ and $w_{i+1}$ in $\DT(W\cup\{p,q\})$ and therefore $p$
 %  and $q$ are adjacent in $\DG^{-}(P,W)$.
\end{proof}

\begin{proposition}
  Let $P$ and $W$ be two point sets in the plane, and
  $n\mathop{:=}\max\{|P|,|W|\}$.  The witness Delaunay graph
  $\DG^{-}(P,W)$ can be computed in $O(n^2)$ time, which is worst-case
  optimal.
\end{proposition}
\begin{proof}
The radial order of the points in $P\setminus\{p\}$ around
each point $p\in P$ can be obtained in overall time $O(n^2)$
\cite{EOS86}, and the Delaunay triangulation $\DT(W)$ can be
constructed in $O(n\log n)$ time \cite{Fo04}. Then for each point
$p\in P$ we can obtain $\DT(W\cup\{p\})$ in additional $O(n)$ time
and traverse the points of $P\setminus\{p\}$ in radial order
within the same time bound, deciding for each one whether it is a
neighbor of $p$ in $\DG^{-}(P,W)$  in constant time, thanks to
Lemma~\ref{lem:neighborsBySector}.
\end{proof}

Although the preceding algorithm is worst-case optimal because the
output may have quadratic size (recall that $\DG^{-}(P,\varnothing)$
is the complete graph $K_{|P|}$), it is interesting to have an
algorithm sensitive to the output size, even if it is more involved.
We show next how to accomplish this.

We first observe that the problem is not interesting if $|P|\leq 1$,
as there are no edges in the graph.  Similarly, if there are no
witnesses, the graph is complete.  In fact, if there is only one
witness, for any two vertices one of the two half-planes defined by
them does not cover a witness; so the graph is again complete.  Thus,
for the remainder of this discussion we assume that $|W|>1$ and
$|P|>1$.
%  \boris{Actually, with a little more work we can get rid of
%   the case of precisely two witnesses as well, which makes the
%   discussion below less ugly, but then I have to add a separate
%   mini-algorithm for that case here, which is comparably ugly.  Or we
%   can just add ``It is not difficult to construct a $O(k+n\log n)$
%   algorithm for the case of two witnesses, so from this point we will
%   assume that $|W|>2$.''  What do you think?}

Given the point sets $P$ and $W$, with $|P|>1$, $|W|>1$, denote by
$V(p)$ the (possibly unbounded polygonal) region of $p\in P$ in the
Voronoi diagram $\Vor(W\cup\{p\})$); note that we allow the
possibility that $p \in W$. Then, we have the following lemma:

\begin{lemma} \label{lem:pseudodisks} The convex polygons $V(p)$, $p
  \in P$ behave as pseudodisks.  More precisely, for distinct points
  $p,q \in P$, $V(p)$ and $V(q)$ are either disjoint or their
  boundaries $\bd V(p)$, $\bd V(q)$ either cross at most twice or overlap
  along a line segment.

  If $W \neq \{p, q\}$, $p \sim q$ in $\DG^{-}(P,W)$ if, and only if,
  $\bd V(p)$ and $\bd V(q)$ meet.  If $W = \{p,q\}$, $p \sim q$ by
  definition of $\DG^{-}(P,W)$.
\end{lemma}
\begin{proof}
  We assume that $W \neq \{ p,q \}$, since the lemma is vacuously true
  otherwise.  By definition, $p \sim q$ in $\DG^{-}(P,W)$ if and only
  if there exists a disk $D_{pq}$ not covering any witnesses, with
  $p,q \in \bd D_{pq}$.  Consider the set of all disks $D$ whose
  bounding circle contains $p$ and $q$; the union of the interiors of
  these disks cover the whole plane, except for a portion of the line
  $pq$.  Since there are other witnesses besides $p$ and $q$ and they
  are not allowed to lie on this line, due to our general position
  assumptions, there is a disk $D$ in this family whose boundary
  passes through $p$, $q$ and another witness $w \in W \setminus \{ p,
  q \}$ and such that $D$ does not cover any witnesses.  The center of
  the resulting disk $D$ is equidistant from $p$, $q$ and the witness
  $w \neq p,q$, and is no closer to any other witnesses.  Hence it is
  a point of $\bd V(p) \cap \bd V(q)$, as claimed.

% (if $W=\{p,q\}$, the statement
%   is true trivially, since then $V(p)$ and $W(q)$ are the
%   non-overlapping regions of $\Vor(W)$)

  Conversely, suppose $c$ is a point of $\bd V(p) \cap \bd V(q)$; let
  $r \mathop{:=} d(c,p) (=d(c,q))$.  By definition of the Voronoi
  regions, the distance from $c$ to the closest witness is $r$.  Hence the
  disk $D_{pq}$ centered at $c$ of radius $r$ covers no witnesses and
  its boundary passes through $p$ and $q$, certifying that $p \sim q$
  in $\DG^{-}(P,W)$.

  The first part of the proof implies that, if $V(p)$ and $V(q)$ meet,
  then their boundaries meet.  To complete the proof of this lemma, it
  is enough to argue that the boundaries meet at most twice or overlap
  in a single segment.  But this is clear, since the intersection of
  $\bd V(p)$ and $\bd V(q)$ lies on the perpendicular bisector of
  $pq$, which is a straight line meeting the boundary of the convex
  polygon $V(p)$, if at all, either in at most two points (where the
  two boundaries properly cross) or overlapping the boundaries of both cells
  in a single connected segment.
\end{proof}

We now use hierarchical representation techniques introduced by
Dobkin and Kirkpatrick \cite{DK1,DK2,DK3}; the properties we need were
summarized by \cite{DHKS}, who refer to \cite{DK1,DK3,M} for the
proofs:

\begin{lemma}[Lemmas 5.2 and 5.3 in \cite{DHKS}]
  \label{lem:hierarchical}
  A three-dimensional polyhedron $R$ with a total of $n$ vertices,
  edges, and faces can be preprocessed in linear time into a data
  structure of linear size that supports the following operations in
  logarithmic time:
  \begin{enumerate}[(a)]
  \item given a directed line $\ell$, find its first point of
    intersection with $R$, and
  \item given a line $\ell$ translating (within in a plane) from
    infinity, find the first point of contact of $R$ and $\ell$.
  \end{enumerate}
\end{lemma}

\begin{theorem} \label{thm:sweepingGhosts} Let $P$ and $W$ be two point sets in the plane, and $n
  \mathop{:=} \max\{|P|,|W|\}$.  The witness Delaunay graph
  $\DG^{-}(P,W)$ can be computed in time $O(k \log n + n \log^2 n)$,
  where $k$ is the number of edges in the graph.
\end{theorem}
\begin{proof}
  As already mentioned, we will assume that there are at least two
  vertices and at least two witnesses.  First suppose that no vertex is
  a witness point; we explain how to remove this assumption below.

  By Lemma~\ref{lem:pseudodisks}, the graph $\DG^{-}(P,W)$ is
  isomorphic to the intersection graph of the set of curves $\{\bd
  V(p) \mid p \in P \}$.  Since any pair of curves cross at most twice
  (or overlap along a segment), it is sufficient to compute their
  arrangement and identify all vertices; a \emph{vertex} is a point of
  crossing of two curves or an endpoint of a segment of overlap.  We
  compute the arrangement by implementing a plane sweep from left to
  right~\cite{4Ms}, \emph{without} representing the curves explicitly,
  since their worst-case combined complexity is easily seen to be
  $\Theta(n^2)$.  We need the following operations:
  \begin{enumerate}[(i)]
  \item For a given $p\in P$, determine the leftmost and rightmost point of
    $\bd V(p)$. [Needed $n$ times, once per $p$.]
  \item For a given pair $p,q \in P$, determine the intersection points of
    $\bd V(p)$ and $\bd V(q)$, or confirm that the curves do not
    meet.  [Needed $O(n+k)$ times, once for every pair of curves
    adjacent along the sweepline.]
  \item For a given point $t(x,y)$ on a vertical line $\ell$, and a
    given $p\in P$, such that $\bd V(p)$ meets $\ell$, determine
    whether $t$ lies in $V(p) \cap \ell$ and, if not, on what side of
    this intersection along the line.  [Needed $O(n \log n)$ times,
    $O(\log n)$ times for each insertion of a new object into the data
    structure maintained by the sweepline; the point $t$ is always the
    leftmost point of a newly discovered region.]
  \end{enumerate}

  With the above three operations in hand, one can carry out a
  standard line sweep, sweeping a plane by a vertical line, say
  left-to-right, detecting appearances, intersections, and
  disappearances of curves and maintaining the order of their
  intersections with the line \emph{without} explicitly computing the
  curves.  It remains to describe how to implement each of the above
  operations to run in logarithmic time.  The claimed running time
  bounds follow.

  Recall the following standard \emph{lifting transformation}: We
  transform a point $p(a,b)\in\reals^2$ to the plane $p^* \colon
  z=2ax+2by-a^2-b^2$ tangent to the \emph{standard paraboloid}
  $z=x^2+y^2$ in $\reals^3$.  The transformation has the following
  property: Given a set $Q$ of points in the plane, consider the set
  $C(Q)$ of all points in space lying on or above the planes of
  $Q^*=\{q^* \mid q\in Q \}$.  This set is an unbounded convex
  polyhedral region whose boundary is a convex monotone surface
  $\pi=\bd C(Q)$.  The surface consists of convex portions
  (\emph{faces}) of the planes of $Q^*$.  The vertical projection of
  $\pi$ to the plane coincides with the Voronoi diagram $\Vor(Q)$ and
  the faces project precisely to Voronoi regions \cite{ES86}.

  We compute and store the polyhedron $C=C(W)$ in a data structure
  supporting operations (a)~and~(b) from Lemma~\ref{lem:hierarchical}.
  We translate the operations~(i)--(iii) to operations on $C$.
  Operation~(iii) involves determining, given $t$, $\ell$, and $p$,
  the location of $t$ along $\ell$, in relation to $\ell \cap V(p)$.
  This can be accomplished in constant time, once we compute $\ell
  \cap V(p)$.  ``Lifting'' the picture to three dimensions, consider
  the line $\ell' \subset p^*$ that projects vertically to $\ell$.
  The desired intersection corresponds to $\ell' \cap C$ in
  $\reals^3$.  This set, in turn, can be computed in $O(\log n)$ time
  by shooting along $\ell'$ in both direction, using
  Lemma~\ref{lem:hierarchical}(a).

  Operation~(ii) again reduces to shooting along a line.  The points
  of $\bd V(p) \cap \bd V(q)$ lie on the bisector $b=b(p,q)$ and their
  corresponding three-dimensional points (i.e., the points of
  intersection of $p^* \cap \pi$ and $q^* \cap \pi$) lie on the line
  $b'=p^* \cap q^*$ that projects vertically to $b$.  Indeed the
  lifted points in question are just the set $b' \cap \pi$ and can be
  computed by two directed-line-shooting queries along $b'$, via
  Lemma~\ref{lem:hierarchical}(a).

  Finally, operation~(i) calls for finding the leftmost (i.e.,
  $x$-minimum) point of $V(p)$; the rightmost point is handled
  similarly.  This point is the projection of the $x$-minimum point of
  $p^* \cap C$ to $xy$-plane.  The latter point is the first point of
  contact of the line $p^* \cap \{ x=c \}$ with $C$, as $c$ varies
  from $-\infty$ to $+\infty$, and so can be identified in logarithmic
  time, by Lemma~\ref{lem:hierarchical}(b).  (It is also possible that
  $V(p)$ does not have a leftmost point---in this case we want to
  compute the infinite ray (or two) bounding $V(p)$ (or $p^* \cap C$
  in three dimensions) and extending to infinity to the left; this is
  needed for properly initializing the state of the sweepline ``at
  infinity.''  This can be done by preprocessing the intersection of
  $C$ with ``the plane'' $x=+\infty$ for line intersection queries and
  intersecting it with $p^*$; being a two-dimensional problem, it is
  easier.)  This concludes our description of the implementation of
  operations~(i)--(iii).

  What modifications are needed if some vertices are also witnesses?
  For such a vertex $p$, $V(p)$ coincides with Voronoi region of $p$
  in $\Vor(W)$; its lifted version is a facet of $C$, which in turn is
  precisely $p^* \cap C$.  Hence the algorithm works as advertised,
  with the additional proviso that the data structure needs to handle
  the possibility that query lines and/or planes might be supporting
  lines/planes of $C$.
\end{proof}

The idea of computing or detecting intersections among a set of
objects by a sweepline algorithm, without explicitly computing the
objects is not new; see, for example, \cite{AAK06,AAS02}.
%% BA I give up.
% \boris{I do
%   not like the implication that we were the ones who invented this.
%   Any earlier/independent citations?}\ferran{No idea.}

\begin{figure}[htbp!]
  \centering
  \includegraphics{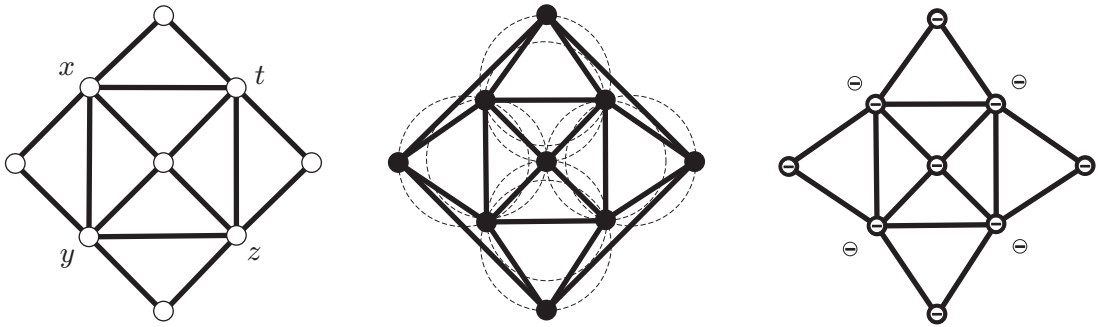}
  \caption{The graph $G$ on the left is not 1-tough, because the
    removal of vertices $x$, $y$, $z$ and $t$ yields five components.
    In the center a supergraph of $G$ is realized as a Delaunay graph;
    from this a witness Delaunay graph realization of $G$ is obtained
    by placing witnesses at the vertices plus four extra witnesses that
    force the removal of the convex hull edges.}
  \label{fig:DelaunayDrawings}
\end{figure}

The characterization of combinatorial graphs that are drawable as
standard Delaunay graphs is a long-standing open problem (see Section
7.3.2 in \cite{Li08}).  Recall that every graph that is realizable as
a Delaunay graph $\DG(P)$ is also a witness Delaunay graph, because
$\DG^{-}(P,P)=\DG(P)$.  In particular, all maximal outerplanar graphs
are realizable, as proved by Dillencourt in \cite{dillencourt2,
  dillencourt3} (better algorithms were later described in
\cite{S94,L97}).

Dillencourt also proved that every Delaunay graph must satisfy
some necessary conditions \cite{dillencourt1}, in particular that
they are always 1-tough (the deletion of any $k$ vertices cannot
produce a graph with more than $k$ components) which he used to
construct some graphs that are not drawable as Delaunay graphs.
For example the graph in Figure~\ref{fig:DelaunayDrawings}, left,
does not admit such a realization because it is not 1-tough.
However it can be realized as a witness Delaunay graph, as shown
in the figure.
% Figure~\ref{fig:DelaunayDrawings}.

%\begin{figure}[htbp!]
%  \centering
%  \includegraphics[scale=1]{DelaunayDrawings.eps}
%  \caption{The graph $G$ on the left is not 1-tough, because the
%    removal of vertices $x$, $y$, $z$ and $t$ yields five components.
%    In the center a supergraph of $G$ is realized as a Delaunay graph;
%    from this a witness Delaunay graph realization of $G$ is obtained
%    placing witnesses at the vertices plus four extra witnesses that
%    force the removal of the convex hull edges.}
%  \label{fig:DelaunayDrawings}
%\end{figure}

Substantial effort has been devoted to drawing trees as proximity
graphs \cite{Li08,BDLL95,BLL96,JLM95,LLMW98,LM03,MS92}.  We prove
next that drawing a tree as witness Delaunay graph is always
possible.

%% TO DO: Change to forests, eventually
% \boris{Changed theorem statement to ``every tree'', though forests
%   also work, but we need to add some more text that I am not willing
%   to add at this point.}
\begin{theorem} % Every acyclic graph can be realized as witness
  Every tree can be realized as witness Delaunay graph $\DG^{-}(P,W)$
  for suitable point sets $P$ and $W$.  The realization can be carried
  out in time linear in the size of the tree, in
  infinite-precision-arithmetic model of computation.
\end{theorem}

\begin{proof} We show that every tree $T=(V,E)$, rooted at a vertex
  $r$, can be drawn as $\DG^{-}(P,W)$, rooted at a given point
  $s$, in such a way that:
  \begin{enumerate}[(a)]
  \item a witness is placed at each vertex, i.e., $P\subseteq W$;
  \item all the vertices $P$ except for the root $s$ are in the
    interior of an axis-parallel square box $B$, and $s$ lies at the
    midpoint of the top side of $B$;
  \item there are disks $D_{ab}$ incident to the endpoints $a, b
    \in P$ of each edge $ab$ corresponding to edges of $T$, empty
    of witnesses and vertices, that certify the Delaunay edges; these
    disks $D_{ab}$ lie inside $B$, except for some disks $D_{sa}$,
    incident to $s$; the disks $D_{sa}$ have their centers inside
    $B$, and can only cross the top side of $B$;
  \item two witnesses are placed at the top-left and the top-right
    corners of $B$.
  \end{enumerate}

%  \begin{figure}[htbp!]
%    \centering
%    \includegraphics{WDGTree2Ferran.eps}%[width=15cm]
%    \caption{The black points are the vertices and the white points,
%      the witnesses.  Rectangle $R$ is dashed.} \label{WDGTree2}
%  \end{figure}

  The proof is by induction on the height $h$ of the tree.  For $h =
  0$, it is obvious as there is only one vertex and no edges; let $B$
  be an arbitrary square with $s$ at the midpoint of its top side.
  Assume this is true for trees with heights up to $k$, $k\geq 0$, and
  let $T = (V,E)$ be a rooted tree of height $k+1$.  Subtrees $T_1 =
  (V_1,E_1), \ldots, T_m = (V_m,E_m)$ of the root have height at
  most $k$, and can be drawn as claimed, in boxes $B_1, \ldots, B_m$,
  by inductive assumption.  By rescaling the boxes, if necessary,
  assume that each has side length $1$.  Place the boxes on a
  horizontal line, in order, $\frac{1}{2}$ apart; refer to
  figure~\ref{WDGTree2}.
  Draw an axis-parallel rectangle $R$ with
  width $\frac{2m-1}{2}$, height $10$ times its width, lower left
  corner at the upper left corner of $B_1$, and lower right corner at
  the upper right corner of $B_m$.  Place $s$ in the middle of the top
  side of $R$. Put three witnesses midway between consecutive boxes
  $B_i$ and $B_{i+1}$, $1 \leq i \leq m-1$, one aligned with the top
  of $B_i$, one with the bottom of $B_i$, and one midway between them.  For $i
  = 1,\cdots,m$, consider a disk $D_{s c_i}$ such that its boundary
  contains $s$ and $c_i$, the root of $T_i$, and such that it is
  tangent to $B_i$. We construct an axis-parallel square box $B$ with
  its upper midpoint at $s$, containing the smaller boxes $B_1$,
  \ldots, $B_m$ and as narrow as possible but yet such that the disks
  $D_{sc_i}$ intersect only its top side. We add a witness at $s$ and
  two witnesses at the two upper corners of the new box $B$ (see
  Figure~\ref{WDGTree2}).
  Notice that the construction creates some collinearities.
  They can be easily removed by slightly perturbing the positions of
  the vertices and witnesses without changing the tree.

  \begin{figure}[htbp!]
    \centering
    \includegraphics{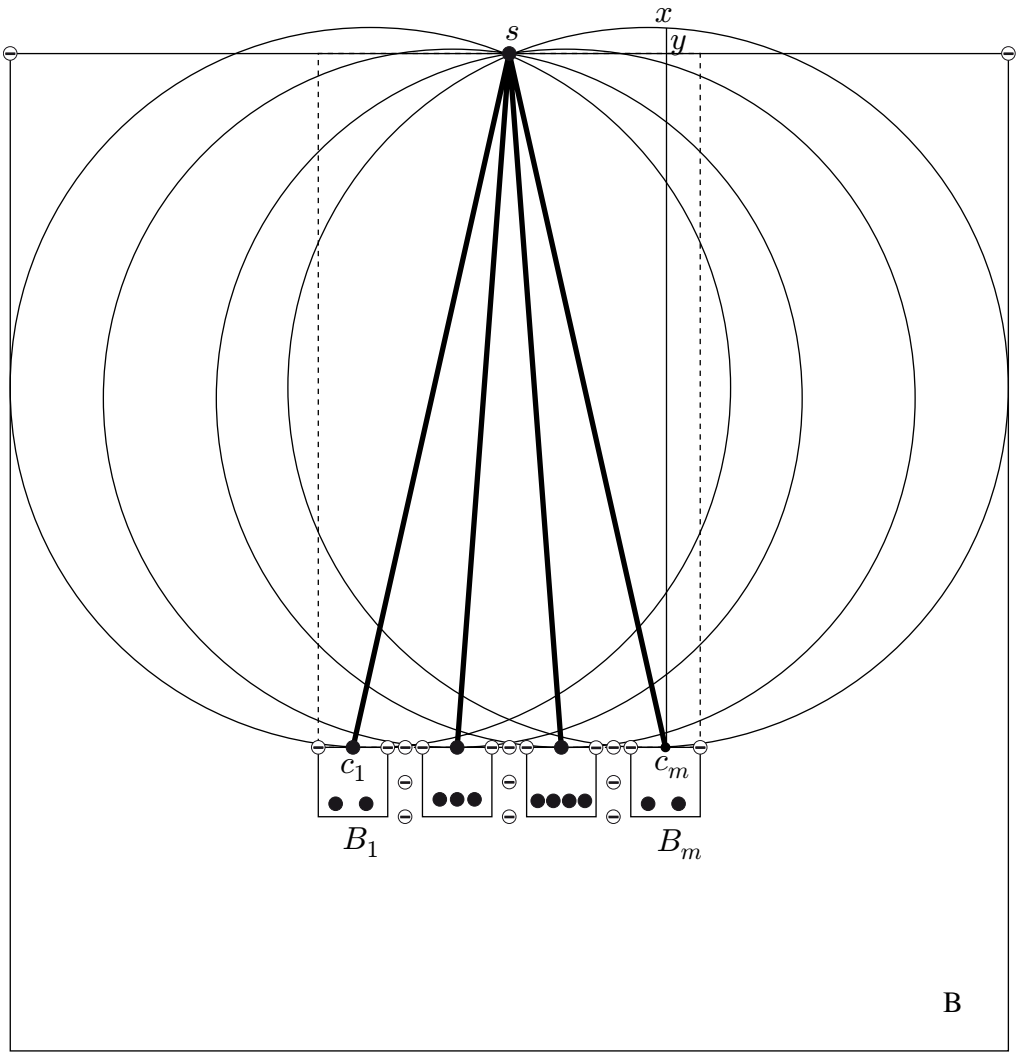}%[width=15cm]
    \caption{The black points are the vertices and the white points,
      the witnesses.  Rectangle $R$ is dashed.} \label{WDGTree2}
  \end{figure}

  To confirm that the construction indeed realizes the tree $T$, we
  first prove a technical assertion: We claim that the sub-boxes $B_1,
  \ldots, B_m$ lie in the lower half of the box $B$.  We prove this by
  induction, for which the base case is vacuously true.  We calculate
  first the side length of $B$.  Let $r$ be the width of $R$.  Let $x
  c_m$ be the diameter of $D_{s c_m}$ incident to $c_m$; refer to
  Figure~\ref{WDGTree2}.  Let $y$ be
  the intersection of $x c_m$ and the horizontal line through $s$.
  We obtain two congruent triangles $\triangle x y s$ and $\triangle s
  y c_m$.  The distance $x y$ is given by $\frac{(r/2 - 1/2)^2}{10
    r}$.  Hence the radius of $D_{s c_m}$ is $\frac{1}{2} \times
  (\frac{(r/2 - 1 /2)^2}{10 r} + 10 r)$. Therefore the width of $B$ is
  $\frac{(r /2 - 1 /2)^2}{10 r} + 10 r + r - 1 $ as the
  centers of $D_{sc_1}$ and $D_{c_ms}$ are $r - 1$ apart.  Now it is
  sufficient to prove that $ \frac{(r /2 - 1 /2)^2}{10 r} + r - 1 < 10
  r$ to show that the smaller boxes $B_1, \cdots, B_m$ are in the
  lower half of $B$.  As $r \geq 1$ and $\frac{(r/2-1/2)^2}{10r} <
  \frac{r}{40}$, we obtain that $ \frac{(r /2 - 1 /2)^2}{10 r} + r - 1
  < 2 r$ and the claim follows.

%  \begin{figure}[htbp!]
%    \centering
%    \includegraphics{WDGTree2Ferran.eps}%[width=15cm]
%    \caption{The black points are the vertices and the white points,
%      the witnesses.  Rectangle $R$ is dashed.} \label{WDGTree2}
%  \end{figure}

  Now we will check that the witness Delaunay graph of the set of
  vertices and witnesses described above is precisely $T$. Conditions
  (a)~to~(d) are clearly fulfilled by construction. By the inductive
  hypothesis, we obtain $T_1, T_2, \ldots, T_m$ as the witness
  Delaunay graphs of the constructions inside boxes $B_1, \ldots,
  B_m$, respectively.  By construction, the witnesses we placed
  outside $B_i$ do not interfere with the Gabriel disks of the edges
  connecting two vertices within $B_i$.  There are no edges $v_i \in
  B_i$, and $v_j \in B_j$, $i \neq j$, because the triples of
  witnesses between the smaller boxes prevent that.  More precisely,
  we know that the edge $v_i v_j$ will cross the segment $w_1 w_2$
  defined by two witnesses $w_1$ and $w_2$, vertically aligned at a
  vertical distance of $\frac{1}{4}$, lying between $B_i$ and
  $B_j$.  If we draw the disk $D_{w_1 w_2}$ with diameter $w_1 w_2$,
  it is empty of vertices by construction. As the edge $v_i v_j$
  crosses it and $v_i$ and $v_j$ are outside $D_{w_1 w_2}$, any disk
  containing $v_i v_j$ contains a witness, for example, by
  Lemma~\ref{lem:adjacency-condition}.

  The roots of $T_1$, $\ldots$, $T_m$ are adjacent to $s$ by
  construction.  It remains to check that $s$ is not adjacent to any
  vertex interior to any of the boxes $B_i$. This is prevented by the
  three witnesses on the top edge of $B_i$.  More precisely, let
  $w_1$, $w_2$, $w_3$, be the three witnesses on the top side of
  $B_i$. We consider the two vertices $s$ and $v$, with $v$ being a
  vertex inside the box $B_i$.  A putative edge $sv$ must cross either the
  segment $w_1 w_2$ or the segment $w_2 w_3$.  Suppose that it
  crosses the segment $w_1 w_2$.  Recall that all the interior
  vertices of $B_i$ are in its lower half, hence the disk $D_{w_1
    w_2}$ with diameter $w_1 w_2$ is empty of vertices. Therefore any
  Delaunay disk $D_{sv}$ must contain either $w_1$ or $w_2$, or both,
  and $sv$ is not an edge of the witness Delaunay graph.
\end{proof}

We note that it might be interesting to investigate how large a grid
one needs to draw a tree as a witness Gabriel graph if the vertices
and witnesses are to be placed at points with integer coordinates.
The above construction made no effort to optimize this quantity.

We conclude this section with a result on the negative side:

\begin{theorem}
  A non-planar bipartite graph cannot be realized as witness Delaunay
  graph $\DG^{-}(P,W)$, for any point sets $P$ and $W$.
\end{theorem}

\begin{proof}
  The proof is by contradiction.  If a realization of a non-planar
  bipartite graph $G$ as a witness Delaunay graph $\DG^-(P,W)$ exists,
  it must contain two crossing edges $p_1 q_1$ and $p_2 q_2$, with
  $p_1$ and $p_2$ belonging to the same part of the bipartite graph,
  and $q_1$ and $q_2$ belonging to the other part.  The vertices
  $p_1$, $p_2$, $q_1$, $q_2$ form a convex quadrilateral $Q$, and we
  may assume without loss of generality that they occur in this order
  along the boundary of $Q$ (see figure~\ref{WDGnoBipartite}).

  \begin{figure}
    \centering
    \includegraphics[scale=1]{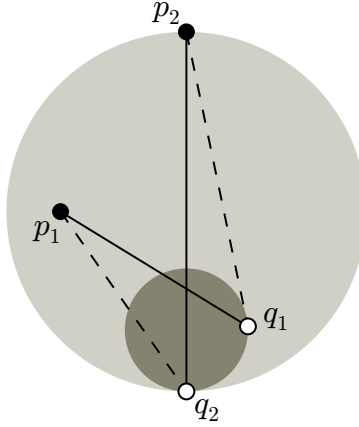}
    \caption{Solid edges are present in the graph, while dashed ones
      may or may not be.}
    \label{WDGnoBipartite}
  \end{figure}
  As $G$ is a bipartite graph, it does not contain $p_1p_2$ or
  $q_1q_2$.  As the sum of the interior angles of a quadrilateral
  equals $360^\circ$, and the vertices are in general position, either
  $\measuredangle p_1 p_2 q_1 + \measuredangle q_1 q_2 p_1 <
  180^\circ$ or $\measuredangle p_2 q_1 q_2 + \measuredangle q_1 q_2
  p_1 < 180^\circ$.  Without loss of generality, suppose
  $\measuredangle p_1 p_2 q_1 + \measuredangle q_1 q_2 p_1 <
  180^\circ$. Then any disk $D$ with $p_2 q_2$ as a chord will contain
  $p_1$, $q_1$, or both.  Let $D_{p_2q_2}$ be the witness-free disk
  certifying the edge $p_2q_2 \in G$.  Since $p_2q_2$ is a chord of
  this disk, it must contain one of $p_1$, $q_1$.  Suppose without
  loss of generality that $D_{p_2q_2}$ contains $q_1$.
  %\boris{We could use my observation here, but it's not needed...}
  By shrinking $D_{p_2q_2}$ by a positive homothety
  with center at $q_2$ until its boundary passes through $q_1$, we
  obtain a disk $D_{q_1q_2} \subseteq D_{p_2q_2}$ not covering any
  witnesses, whose boundary contains $q_1,q_2$, contradicting $q_1q_2
  \not\in G$.
% As $q_1$
%   and $q_2$ are not adjacent, all the Delaunay disks of $q_1 q_2$
%   contain a witness.  Therefore, as $q_1$ and $q_2$ are in $D_{pq}$,
%   at least one Delaunay disk of $q_1q_2$ is in $D_{pq}$, a
%   contradiction.
\end{proof}

%%%%%%%%%%%%%%%%%%%%%%%%%%%%%%%%%%%%%%%%%%%%%%%%%%%%%%%%%%%%%%%%%%%%%%%%%%%%%%
\section{Square graphs}\label{section:squares}
In this section we use the term \emph{square graph} as short for the
square graph of a point set $P$ (the \emph{vertices}) with respect to a second
point set $W$ (the \emph{witnesses}); we recall that two points $x,y\in P$ are
adjacent in the graph $\SG^{+}(P,W)$ if and only if there is an axis-aligned square with $x$ and $y$ on
its boundary whose interior contains some witness point $q\in W$. As mentioned in the introduction, this
is the positive witness version on the Delaunay graph for the
$L_{\infty}$ metric.  We assume that no two distinct points in $P\cup W$ have equal $x$- or
$y$-coordinates and let $n \mathop{:=} \max\{|P|,|W|\}$.  (In this
section, we do not require that no three points be collinear.)  We denote by
$E$ the edge set of the graph; we partition $E$ into $E^+$ and $E^-$
according to the slope sign of the edges when drawn as segments.

First, a simple geometric observation:
\begin{observation} \label{obs:shrunken-square}
  If $R$ is a closed square containing points $p$ and $q$ then there exists a
  square $R_{pq}\subset R$ whose boundary passes through $p$ and $q$.
\end{observation}
\begin{proof}
  Let $c$ be the center of $R$.  Shrink~$R$ while keeping its center
  at $c$ until it is about to lose $p$ or $q$.  Let the resulting square
  be $R'$.  Without loss of generality, let $p \in \bd R'$.
  Shrink~$R'$ by a homothety with center $p$ until it is about to lose
  $q$.  The result is the desired square~$R_{pq}$.
\end{proof}

The isothetic rectangle (\emph{box}) defined by two points $p,q$ in
the plane is denoted $B(p,q)$.  For an edge $e=pq$ we also write
$B(e)$ instead of $B(p,q)$.  Every edge $e$, say in $E^+$, defines
four regions in the plane as in Figure~\ref{fig:cornerBay}, that we call
\emph{corners} and \emph{bays}.
A corner is a closed set while a bay is an open set.

\begin{figure}[htbp!]
  \centering
  \includegraphics[scale=1]{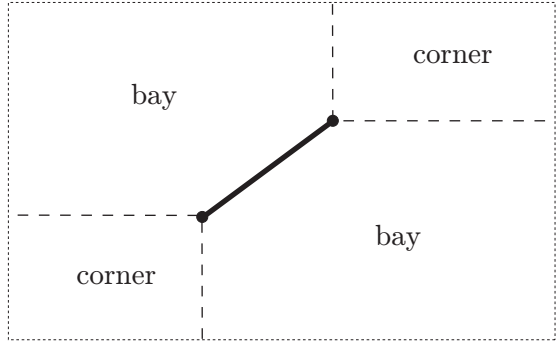}
  \caption{Corners and bays.}
  \label{fig:cornerBay}
\end{figure}

If $p,q\in P$, then $pq\in E$ if and only if the
union of the bays of $pq$ contains some witness or, equivalently,
if and only if $W$ is not contained in the union of the corners.
In particular, the placement of just two witnesses, for example,
just outside and
very close to the top corners of an axis-aligned rectangle
enclosing $P$ suffices to yield a complete graph, because every
upper bay would contain a witness.  Also, as the
bays associated with a pair of points $p,q$ cover the vertical open strip
delimited by the lines $x=x(p)$ and $x=x(q)$ and the horizontal
open strip delimited by the lines $y=y(p)$ and $y=y(q)$, we deduce the
following useful fact.

\begin{observation}
  \label{obs:strips}
  If there is a witness point $w\in W$ such that $x(w)$ is between
  $x(p)$ and $x(q)$ or $y(w)$ is between $y(p)$ and $y(q)$,
  then $pq$ is an edge of $\SG^{+}(P,W)$.
\end{observation}

Computing how many witnesses are contained in quadrant~I and
quadrant~IV for every $p\in P$ can be carried out in overall $O(n\log
n)$ time with a line sweep from right to left, and by keeping the set
of witnesses already encountered stored in a balanced search tree,
sorted by the $y$-coordinate; a sweep in the opposite direction
handles the remaining two quadrants.  After that every pair of points
$p,q\in P$ can be checked for adjacency in constant time and therefore
the square graph $\SG^{+}(P,W)$ can be computed in $O(n^2)$ time,
which is worst-case optimal.  We describe next an output-sensitive
algorithm.

\begin{theorem}
  \label{thm:sensitiveSquare}
  Let $P$ and $W$ two point sets in the plane, and $n \mathop{:=}
  \max\{|P|,|W|\}$.  The square graph $\SG^{+}(P,W)$ can be computed
  in optimal $O(k+n \log n)$ time, where $k$ is the number of edges in
  $\SG^{+}(P,W)$.
\end{theorem}

\begin{proof}
We first detect all pairs of points $p,q\in P$ such that the open
strip bounded by the vertical lines through these points covers
some witness, making them adjacent in the graph
(Figure~\ref{fig:fourthQuadrantCS}, left).  For this, it
suffices to consider the projection $z^*$ of all the points $z\in
P\cup W$ onto the $x$-axis. Once the projections are sorted, it is
clear that for every $p\in P$, if $w$ is the first witness such
that $x(p)<x(w)$, we can simply list all the $q\in P$ such that
$x(w)<x(q)$. After the $O(n\log n)$ sorting step, a simple scan
gets every adjacency listed once, and the global cost is
proportional to their number. On the other hand, in $O(n)$ time
after sorting, we can also store for each point $p\in P$ the
number of witness projections to the right of $p^*$. This will
later allow us to detect in constant time whether there is a
witness in the vertical strip defined by $p,q\in P$.

\begin{figure}[htbp!]
  \centering
  \includegraphics[scale=1]{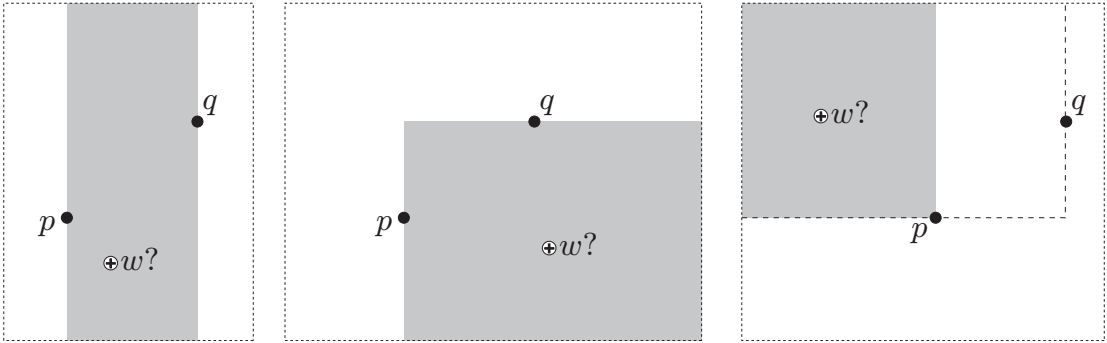}
  \caption{Illustrating the proof of Theorem~\ref{thm:sensitiveSquare}.}
  \label{fig:fourthQuadrantCS}
\end{figure}

  Next we explain how to find the pairs of points $p,q\in P$, with $x(p)<x(q)$, such that the
  slope of the segment $pq$ is positive, $p\sim q$ in $\SG^{+}(P,W)$, and the adjacency has
  not been reported in the previous
  step. The remaining case is handled in a symmetric
  manner.

Sweep from right to left with a vertical line $\ell$, and maintain
the lowest witness $w_R\in W$ to the right of $\ell$ and the
highest witness $w_L\in W$ to the left of $\ell$ . In addition,
maintain a $y$-sorted list $L$ of the points of $P$ to the right
of $\ell$.  When the sweep line finds a point $p\in P$, we report
all the points $q\in L$ that are above $w_R$ and $p$, by a simple
linear scan of the list from $p$, omitting those that have
reported as adjacent to $p$ in the preceding step that checked the
vertical strip between them (Figure~\ref{fig:fourthQuadrantCS},
center).

If both $w_L$ and $w_R$ are above $p$, we additionally report all
the points $q\in L$ that are above $p$ and  below $w_R$,
performing a second linear scan of the list from $p$, again
omitting those that have reported as adjacent to $p$ in the first
step (Figure~\ref{fig:fourthQuadrantCS}, right).

The involved costs are $\Theta(1)$ per edge found, $\Theta(\log
n)$ to insert $p$ into $L$, and $\Theta(1)$ to update $w_L$ and
$w_R$  when a witness from $W$ is encountered by the sweep line.

This process must be repeated from left to right for edges with
negative slope.  Overall, all edges will be found and each
one reported exactly once, which proves that the graph can be
computed in $O(k+n \log n)$ time, as claimed.

Let us now show that this is optimal. The lower bound $\Omega(k)$ is
obvious. To see the $\Omega(n \log n)$ part, we use a reduction from
the \textsc{uniqueness} problem: ``Given $n$ positive integers, decide
whether all of them are distinct.'' which is known to have an $\Omega(n
\log n)$ lower bound in the algebraic computation tree model
\cite{Y91}.

Now, given positive integers
$S=\{x_1, ..., x_n\}$, consider the point set $P=\{p_1,\ldots,p_n\}$,
with $p_i=(x_i-\frac{i}{10n},x_i+\frac{i}{10n})$.  It is easily
checked that $P$ is a set of points near the line $x=y$, such that, as
long as $x_i\neq x_j$, for $i\neq j$, the slope of the segment
$p_ip_j$ is positive.  However, if $x_i = x_j$, for some $i \neq j$,
the slope of $p_ip_j$ is $-1$.  In particular, $SG^+(P,\{(0,0)\})$ has
no edges if and only if all numbers in $S$ are distinct.  This
completes the description of a linear-time reduction from
\textsc{uniqueness} to the computation of the square graph, hence
proving the claimed complexity lower bound.
\end{proof}

Before describing the combinatorial structure of square graphs, we
recall some well-known definitions.

Given points $a, b\in \reals^d$, with $a=(a_1,\dots,a_d)$ and
$b=(b_1,\dots,b_d)$ we say that $a$ \emph{dominates} $b$ (denoted as
$a \geq b$ or $b \leq a$) when
$a_i\ge b_i$ for $i=1,\dots,d$.  Given a partially ordered set
 $\P=(X,{\le}_{\P})$, a $d$\emph{-dominance realization}
of $\P$ is a function $f:X\rightarrow\reals^d$
such that $x\le_{\P} y$ if and only if $f(x) \mathbin{\leq} f(y)$,
for all $x,y\in X$.

The smallest $d$ such that $\P=(X,\le)$ admits a $d$-dominance
realization is called the \emph{dimension} of the partial order $\P$.
Equivalently, $d$ is the smallest integer such that $\P$ is the
intersection of $d$ total orders that are extensions of $\P$.  The
concept of dimension was introduced, and the equivalence of the
definitions proved, in \cite{DM41}.

The undirected graphs underlying partial orders (i.e., for distinct $x,y$, $x\sim y$
when $x\le y$ or $y\le x$) are called \emph{comparability graphs}.
It has been proved (see section 6.2 in \cite{survey99}) that any
two partial orders whose underlying comparability graphs are the
same must have the same dimension, and therefore we can call this
number the \emph{dimension of the comparability graph}.  The
comparability graphs corresponding to two-dimensional partial
orders are called \emph{permutation graphs} (this name arose in a
different context yet equivalence was established).

We are now ready for our main result in this section, a complete
characterization of square graphs:

\begin{theorem} \label{thm:squareStructural} A combinatorial graph
  $G=(V,E)$ can be realized as a square graph $\SG^{+}(P,W)$ for
  suitable point sets $P$ and $W$ in the plane if and only it is a permutation graph.  Moreover, any square graph
  can be realized using at most one witness.
\end{theorem}

\begin{proof}
  Let $G=\SG^{+}(P,W)$ be a square graph and $G'$ its complement. Recall that we assume that
  no two distinct points in $P\cup W$ have equal $x$- or
  $y$-coordinates.  Suppose first that $P \cap W = \emptyset$, i.e.,
  no vertex is also a witness.  We will remove this assumption below.

  Draw a vertical line and a horizontal line through each witness,
  partitioning the plane into open \emph{boxes}.  From
  Observation~\ref{obs:strips} we know that no edge of $G'$ crosses
  any of these lines.

  Consider such a box $B$.  By construction, the vertical open strip
  containing $B$ covers no witnesses and the same is true of the
  horizontal open strip containing $B$.  We partition the complement of the
  union of these strips into four closed \emph{quadrants of $B$},
  numbered I~through~IV (if $B$ is unbounded in one or more
  directions, we simply treat two or more of the quadrants as empty
  sets), refer to Figure~\ref{fig:SquaresQuadrants}.

  \begin{figure}[htbp!]
    \centering
    \includegraphics{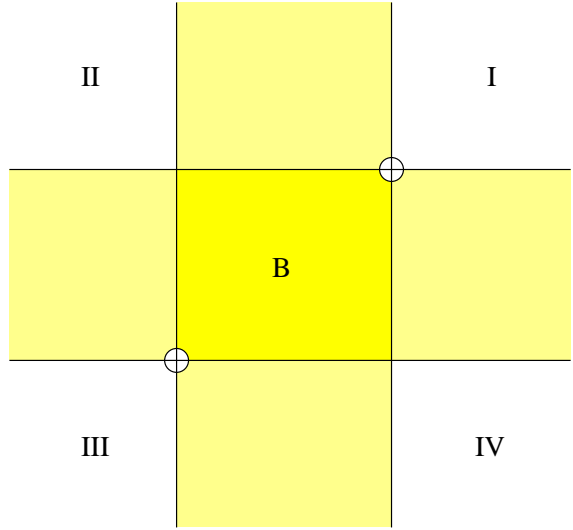}
    \caption{Illustrating the proof of Theorem~\ref{thm:squareStructural}.
      % The open box $B$ is shaded yellow.
      The two open strips are light yellow.  Their complement consists
      of the four quadrants numbered I~through~IV.}
    \label{fig:SquaresQuadrants}
  \end{figure}

  Putting $P_B \mathop{:=} B \cap P$, let $G_B=(P_B,E_B)$ be the
  subgraph of $\SG^{+}(P,W)$ induced on $P_B$, and let
  $G'_B=(P_B,E'_B)$ be its complement.  If there is a pair of adjacent
  quadrants of $B$ each containing a witness, the graph $G_B$ is
  complete and its complement $G'_B$ is the empty graph on $P_B$.  The
  empty graph is certainly a two-dimensional comparability graph, as
  it suffices to take a sequence of points with increasing abscissae
  and decreasing ordinates.

  There is only one case remaining: there are witnesses only in one
  pair of opposite quadrants of $B$ (one of these opposite quadrants
  could be empty of witnesses); without loss of generality, these
  quadrants are I~and~III.  Then $p,q\in P_B$ define an edge of $G_B$
  if, and only if, the slope of the segment $pq$ is negative.  Hence
  $G'_B$ is the comparability graph underlying the dominance relation
  for $P_B$ with the current system of coordinates.

  Therefore, we have proved that the complement of $\SG^{+}(P,W)$ is
  the disjoint union of permutation graphs which itself is a
  permutation graph, if no vertex is also a witness. As the complement of a permutation graph is also a permutation graph (see \cite{PermutationGraphs}), we have proved that $\SG^{+}(P,W)$ is a permutation graph as well.

  Now suppose $P \cap W \neq \emptyset$.  The above argument applies
  verbatim to the subgraph of $G'$ induced on $P \setminus W$, i.e.,
  to the non-adjacencies between non-witness vertices.  Let $q \in P
  \cap W$.  Let the superbox $H$ of $q$ be the smallest open box
  enclosing the four open boxes (which we call $B_{\mathit{I}}, B_{\mathit{II}}, B_{\mathit{III}},
  B_{\mathit{IV}}$ according to their position around $q$) adjacent to
  $q$; refer to Figure~\ref{fig:SquaresQuadrants2}.
  Let $\bar H$ be the closure of $H$.  Using
  Observation~\ref{obs:strips}, we conclude that $q$ is adjacent in
  $G$ to every vertex outside of $\bar H$.  In particular, in $G'$,
  all neighbors of $q$ lie in $\bar H$.

  \begin{figure}[htbp!]
    \centering
    \includegraphics{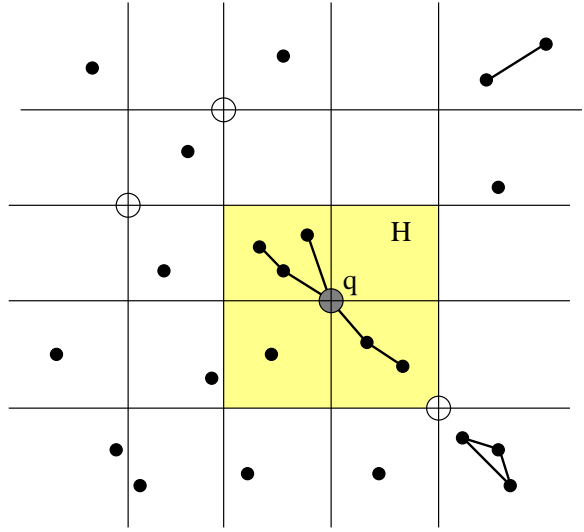}
    \caption{% Illustrating $G'$ in the proof of Theorem~\ref{thm:squareStructural}.
      The superbox $H$ is shaded yellow, $q \in P \cap W$ is grey,
      points of $P \setminus W$ are black, and witnesses of $W \setminus P$ are white with a cross.}
    \label{fig:SquaresQuadrants2}
  \end{figure}

  We first consider the special case $W=\{q\}$.  As argued above,
  disregarding $q$, the complement of $G$ is the disjoint union of at
  most four permutation graphs, one for each of the boxes surrounding
  $q$.  By definition of a square graph, $q$ is not adjacent to
  anything in $G$ and hence $G'$ is formed by taking the disjoint
  union of four or fewer permutation graphs and adding a vertex
  adjacent to all other vertices.  We argue that then $G'$ is a single
  permutation graph and hence $G$ is a permutation graph as well.  Indeed, form a 4-by-4 grid in the plane and draw
  each of the permutation graphs of $G'$ in the diagonal boxes of the grid,
  top-left to bottom-right, so that each coincides with the
  comparability graph of their $xy$-dominance relation.  There are
  no dominance relations between the diagonal boxes, so we have a
  realization of their disjoint union.  Now place the vertex
  corresponding to $q$ below and to the left of the grid, completing
  the realization of $G'$ as the comparability graph of a
  2-dimensional dominance relation.

  For the remainder of this proof, we assume that $q$ is not the only
  witness.

  %% UGLY and AWKWARD. Attempting a rewrite.  --BA
  % We first prove the following claim: in $G'$, $q$ is
  % adjacent to at most two other vertex-witnesses, which must lie at
  % diagonally opposite corners of $H=H(q)$.  Indeed, since $q$ can only
  % be adjacent to vertices in $\bar H$, and no two witnesses can share
  % a vertical or horizontal line, vertex-witnesses other than $q$ in
  % $\bar H$ can lie only at its corners, so there are only three
  % possibilities: $\bar H$ can contain none, one, or two
  % vertex-witnesses besides $q$.  If there are two, the
  % vertex-witnesses must lie at opposite corners.  In all cases, by
  % definition of a square graph, when considering the effect of only
  % these vertex-witnesses, $q$ is not adjacent to these
  % vertex-witnesses in $G$ and therefore is adjacent to them in $G'$.
  % Presence of more witnesses may add some edges in $G$ and therefore
  % eliminate some edges in $G'$.  This completes the proof of the
  % claim.  In particular, the subgraph $G'_{P\cap W}$ induced by
  % vertex-witnesses in $G'$ is a collection of paths, and each path is
  % either a chain or an antichain in the dominance relation.

  We first prove the following claim: the subgraph $G'_{P\cap W}$
  induced by vertex-witnesses in $G'$ is a collection of paths, and
  each path is a chain or an antichain in the dominance
  relation on $P$,
  i.e., the vertices along the path either have increasing $x$- and
  $y$-coordinates, or increasing $x$- and decreasing $y$-coordinates.

  Indeed, by construction and by our assumption that no two
  vertex-witnesses share $x$- or $y$-coordinates, $\bar H$ can contain
  at most two vertex-witnesses besides $q$.  Any such vertex-witness
  must lie at a corner of $\bar H$, and if there are two of them, they
  must occupy diagonally opposite corners of $\bar H$.  As already
  observed, $q$ is not adjacent in $G'$ to any vertex outside $\bar
  H$, hence it has degree at most two in $G'_{P\cap W}$ and if it does
  has two neighbors, they form a chain or an antichain with $q$ in the
  dominance relation on $P$.  The claim easily follows.

  Let $G_H$ ($G'_H$) be the subgraph of $G$ (respectively, $G'$)
  induced by the vertices in the open box $H$, i.e., by $q$ and the
  vertices in $B_{\mathit{I}}, \ldots, B_{\mathit{IV}}$.  Consider the
  union of the open vertical and horizontal strips defined by $H$; its
  complement is a union of (at most) four quadrants, which we refer to
  as \emph{the quadrants of $H$} and number in the usual manner,
  I~through~IV.  The quadrants contain all the witnesses besides~$q$,
  and we have assumed there exists at least one such witness.  As
  above, if there is a pair of adjacent quadrants containing
  witnesses, $G_H$ is complete and therefore $G'_H$ is the empty
  graph; in fact it is easy to check that $G_{\bar H}$ is complete,
  $G'_{\bar H}$ is empty, and therefore $q$ is an isolated vertex in
  $G'$ in this case.

  There remains the case that there is a pair of opposite quadrants,
  say I~and~III, one or both of which contain a witness.  Then
  $G_{B_{\mathit{II}}}$ and $G_{B_{\mathit{IV}}}$ are complete graphs
  and therefore $G'_{B_{\mathit{II}}}$ and $G'_{B_{\mathit{IV}}}$ are
  empty graphs; $q$ is not adjacent in $G'$ to any vertex in these two
  boxes.  On the other hand, $G'_{B_{\mathit{I}}}$ and
  $G'_{B_{\mathit{III}}}$ represent the dominance relation in
  $B_{\mathit{I}}$ and $B_{\mathit{III}}$, respectively, and $q$ is
  adjacent to every vertex in those two boxes.  If either of the
  quadrant~I or quadrant~III corners of $H$ is a vertex-witness, they
  are the neighbors of $q$ along its path in $G'_{P \cap W}$.

  To summarize, $G'$ decomposes into disjoint subgraphs of the
  form
  \[\langle K_0, q_1, K_1, q_2, \ldots, q_\ell, K_\ell \rangle,\]
  with $\ell \geq 0$, where each $K_i$ is a permutation graph
  (possibly with no vertices), each $q_i$ is a vertex-witness,
  $q_1,q_2,\ldots,q_\ell$ is a simple path in $G'$ (in fact, it is a chain or
  an antichain in the 2-dimensional dominance relation on $P$) and
  all vertices of $K_i$ are adjacent to $q_i$ and $q_{i+1}$, for
  $0<i<\ell$.
  % \ferran{I am lost. Let us take $a$ and $b$ in $K_1$.
  %   If $ab$ has positive slope, $q_3$ makes them adjacent in $G$; if negative
  %   $q_2$ makes them adjacent in $G$. Same reason $a$ is adjacent to $q_1$ in $G$.}
  %
  Vertices of $K_0$ are adjacent to $q_1$, vertices of $K_\ell$ are adjacent
  only to $q_\ell$.  There are no adjacencies between vertices of
  different $K_i$'s.

  The first part of the proof of the theorem is complete once we argue
  that each $\langle K_0, q_1, \ldots \rangle$ is isomorphic to the
  comparability graph of the dominance relation of some
  two-dimensional set of points; such a realization in depicted in
  figure~\ref{fig:chains} (Notice that this realization uses a set of
  points unrelated to $P$); as hence follows that the complement of all $\langle K_0, q_1, \ldots \rangle$ is a comparability graph.
\begin{figure}
  \centering
  \includegraphics{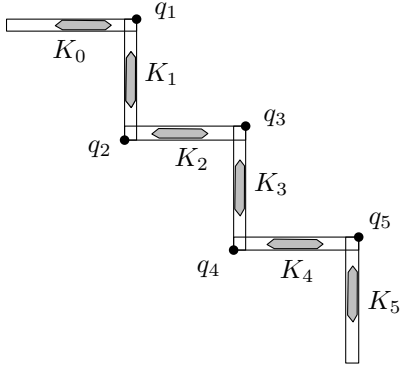}
  \caption{In most general case, a connected component of $G'$ is a
    comparability graph of a two-dimensional dominance relation.  The
    shaded regions represent two-dimensional dominance realizations of graphs
    $K_i$.}
  \label{fig:chains}
\end{figure}

\bigskip

Conversely, let $G=(V,E)$ be a comparability graph of dimension 2. 
Consider a set~$P$ of points
in the plane with no repeated coordinate values, whose dominance graph
(for $(x,y)$-coordinates) is isomorphic to $G$.  Points $p,q\in P$
are adjacent in $G$ if and only if $pq$ has positive slope.  If now
we place a single witness point $w$ to left and below~$P$, an edge
appears in $\SG^{+}(P,\{w\})$ if and only if the slope of $pq$ is
negative, so $\SG^{+}(P,\{w\})\cong G$.
\end{proof}

%%% END OF NEW PROOF

Recognizing whether a
combinatorial graph $G=(V,E)$ is a permutation graph can be done in
time $O(|V|+|E|)$ \cite{permutation99}. Combining this result with the
preceding theorem, we immediately obtain:

\begin{corollary}
  One can decide in $O(|V|+|E|)$  time whether a given
  combinatorial graph $G=(V,E)$ can be realized as a square graph
  $\SG^{+}(P,W)$ for some point sets $P,W$ in the plane.  If a
  realization exists, it can be constructed within the same time
  bounds.
\end{corollary}

\section{Size of stabbing sets}\label{section:stabbing}

Consider a point set $P$ in the plane.  Let
$\S$ be a family of geometric objects with nonempty
interiors, each one associated to a finite subset of $P$.  We say
that a point $w$ \emph{stabs} an object $Q\in \S$ if $w$
lies in the interior of $Q$. In this section we consider the
problem of how many points are required to stab all the elements
of $\S$, which we denote by $st_{\S}(P)$, and how
large this number can be when all the point sets with $|P|=n$ are
considered. We denote this extremal value by
$st_{\S}(n)=\max_{|P|=n} st_{\S}(P)$. We will see
that these problems can be rephrased in terms of witness graphs,
and therefore the results from the previous sections be used for
their study. On the other hand, the more natural and interesting formulation
is in terms of Voronoi discrimination, a description that we present in
Subsection \ref{subsection:discrimination}.

Similar problems have already been considered for the family of all
convex polygons with vertices from among the points of $P$.  In
particular, if $\mathcal{T}$ and $\mathcal{Q}$ are the family of
triangles and convex quadrilaterals, respectively, with vertices in
the given point set, it has been proved that $st_{\mathcal{T}}(n)=2n-5$
and that $st_{\mathcal{Q}}(n)=2n-o(n)$ \cite{KM88,CKU99,SU07}. Those
families of shapes are finite, while the ones we consider here are
infinite and continuous: the disks and the isothetic squares whose
boundary contains two points from $P$. For these problems the stabbing
set can be viewed as a witness set that yields a specific type of
corresponding witness graph, a connection that allows us to use the
preceding results and that we make precise below.

\subsection{Stabbing disks}\label{subsection:stabbingCircles}

Let $P$ be a set of $n$ points, %no four concyclic,
and let $\D$ be the set of disks whose boundary contains at least two
points from $P$.  If $p,q\in P$, a set of points $W$ stabs every disk
with $p$ and $q$ on its boundary if and only if $pq$ is not an edge of
$\DG^-(P,W)$.  In other words,
\[
  st_{\D}(n)=\max_{|P|=n} st_{\D}(P)=\max_{|P|=n} \min
  \{|W| : \DG^-(P,W)=\varnothing\}.
\]
% \boris{Ferran, we do not need general, position, but as Muriel pointed
%   out, the definition of $\DG$ seems to require general position, as
%   stated in our very own words...  Of course, this is really
%   hair-splitting, but...}
% where the extrema are evaluated over all pairs $(P,W)$ of sets with
% $P\cup W$ in general position.
% \boris{Stuck general position here.
%   Ugly, but I am worried that without this condition the quantity we
%   are trying to define has a different value.}\ferran{Really? Now we have defined 'to stab' in a very
%   convenient way, see the first three lines of Section 4.}

% \boris{Do we need to add ``no three on a line'' below?}
% \ferran{I don't see why. In fact I don't even see why we want to
%   prevent concyclicity}
% \Boris{I think I decided that I agree with Ferran.  Muriel, can you
%   please check that no non-degeneracy assumption is needed in this
%   section's arguments, when you are re-reading this?  I removed the
%   assumptions from the text already.  I hope.}

\begin{lemma}
  \label{lem:disjointCircles}
  Let $P$ be a set of $n$ points and let $\D'$
  be a set of disks with pairwise disjoint interiors, such that the
  boundary of any of them contains two points from $P$. Then,
  $|\D'|\le n$, which is tight.
\end{lemma}
\begin{proof}
A point $p\in P$ can lie on the boundary of at most two interiorly
disjoint disks from $\D'$, and which would necessarily be
tangent at $p$. As every bounding circle contains two points from
$P$, $|\D'|\leq n$. This bound is achievable, for
example, in a necklace of disks in which each one touches its
two neighbors, with $P$ being the set of contact points.
\end{proof}

\begin{lemma} \label{lem:delaunayEdgesSuffice}
%  OLD:
%  Let $P$ be a set of $n$ points, no four on a circle.  If, for every
%  edge $pq$ of $\DT(P)$, $pq \notin \DG^{-}(P,W)$, then
%  $\DG^-(P,W)=\varnothing$.
  Let $P$ be a set of $n$ points. %no four on a circle.
  If none of the edges of $\DT(P)$ are in $\DG^{-}(P,W)$,
  then $\DG^-(P,W)=\varnothing$.
\end{lemma}
\begin{proof}
  Let $p$ and $q$ be two points from $P$ and let $C$ be any circle
  through them.  If $p$ and $q$ are neighbors in $\DT(P)$, we know by
  hypothesis that $p\not\sim q$ in $\DG^-(P,W)$.  If $p$ and $q$ are
  not adjacent in $\DT(P)$, then there is at least one point from $P$
  in the interior of the disk $D$ bounded by $C$, and we can find a
  disk $D'$ contained in $D$, tangent to $C$ at $p$, that has a second
  point, say $p'$, from $P$ on its boundary but none interior.  ($D'$
  can be obtained, for example, by shrinking $D$ with center $p$ until
  the moment it contains no point of $P$ in its interior.)
  %\boris{use observation?}
  Therefore $pp' \in \DT(P)$.  Since $pp'$ is not an edge of
  $\DG^{-}(P,W)$, $D'$ must contain a witness point.  This witness
  stabs $D$ as well, therefore $p\not\sim q$ in $\DG^-(P,W)$.
\end{proof}

The preceding lemma implies that to stab all disks whose boundaries
contain pairs of points from $P$, it is enough to stab only the disks
corresponding to pairs of Delaunay neighbors.
% \boris{Should we add something like ``However, we need to make sure that together with
%the added witnesses the points are in general position.''?}
%This can easily be done by placing a witness point on each Delaunay edge,
%or two very close to the edge on opposite sides if we prefer to avoid
%collinearities, \boris{Whether or not we ``avoid collinearities''
%  fundamentally changes the problem.  Can't just sweep this under the
%  rug...}  yielding roughly a total of $3n$ and $6n$ witnesses,
%respectively. We show next a better upper bound.
This can easily be done by placing a witness point very close to the midpoint of each Delaunay edge, the witness point being external for the convex hull edges, yielding roughly a total of $3n$ witnesses. We show next a better upper bound.

\begin{theorem}  \label{thm:stabbingCircles}
For $n\geq 2$, $n\le st_{\D}(n)\le 2n-2.$
\end{theorem}
\begin{proof}
  The lower bound comes from the existence of sets, as shown in
  Lemma~\ref{lem:disjointCircles}, that admit $n$ disks with pairwise
  disjoint interiors, each containing two of the points on its
  boundary, because each disk requires a distinct stabbing witness.
  For the upper bound we place a witness $p_T$ inside each Delaunay
  triangle $T$ in $\DT(P)$, in such a way that $p_T$ sees every side
  of $T$ with an angle greater than $\pi/2$ (for example, one may
  place $p_T$ on an internal height, very close to its foot).  We also
  place a witness for every edge of the convex hull, external and very
  close to its midpoint. In this way every disk having a Delaunay edge
  from $\DT(P)$ as a chord will be stabbed at least on one side of the
  edge.  If the size of the convex hull is $h$, the number of
  triangles in $\DT(P)$ is $2n-h-2$, therefore the total number of
  witnesses we have placed is $2n-2$, as claimed.
\end{proof}

%\begin{figure}[htbp!]
%\centering
%\includegraphics[width=12cm]{WDGverticesConvex2.eps}
%\caption{The black points, triangles and squares are the vertices
%and the white points, the witnesses.} \label{WDGverticesConvex2}
%\end{figure}

We don't believe the upper bound on the previous theorem to be
tight. We have obtained at least a better bound for points in
convex position:

\begin{proposition} \label{prop:stabbingCirclesConvex}Let $P$ be a set of $n$ points in convex position,
  then $st_{\D}(P)\le \frac{4}{3}n$.  In other words, a suitable set
  $W$ of at most $\frac{4}{3} n$ witnesses is always sufficient to have
  $\DG^-(P,W)=\varnothing$.
\end{proposition}
\begin{proof}
  Recall from Lemma~\ref{lem:delaunayEdgesSuffice} that to eliminate
  all edges in a witness Delaunay graph of a set of points $P$, it is
  sufficient to eliminate the edges of the Delaunay triangulation of
  $P$.  Color the vertices of the Delaunay triangulation with three
  colors \cite{3colorsFisk}, white, gray and black. Pick the color that covers the largest
  number of vertices, suppose this is color is gray. For each vertex $v$ of
  color black or white, and its two incident edges on the convex hull $va$ and
  $vb$, put a witness $w_1$ outside of $\CH(P)$ very close to $v$ and
  $va$, and another witness $w_2$, outside of $\CH(P)$, and very close
  to $v$ and $vb$ (see Figure~\ref{WDGverticesConvex4}).
  \begin{figure}[htbp!]
    \centerline{\includegraphics{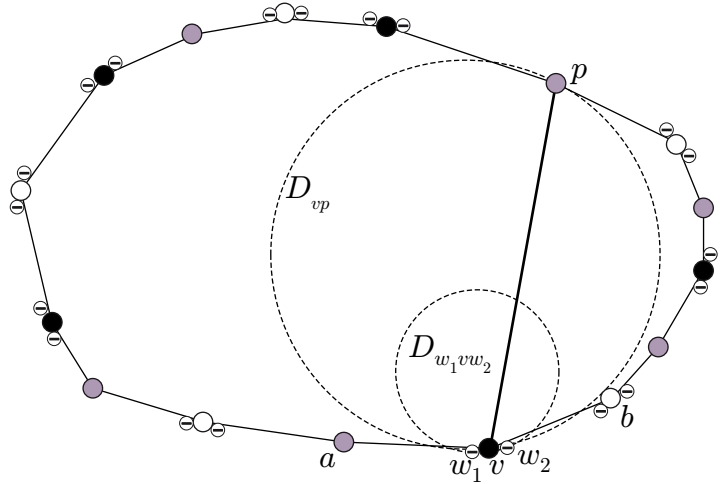}}
    \caption{Illustration for the proof of Proposition~\ref{prop:stabbingCirclesConvex}}
    \label{WDGverticesConvex4}
  \end{figure}
  The witnesses $w_1$ and $w_2$ are close enough to $v$ so that the
  disk $D_{w_1 v w_2}$ defined by $w_1$, $v$, and $w_2$ is empty of
  vertices.
  Consider any Delaunay edge $vp$ incident to $v$.  By construction,
  $p$ is outside of $D_{w_1vw_2}$ and $vp$ intersects the interior of
  $D_{w_1vw_2}$.  Therefore $\measuredangle vw_1p + \measuredangle
  vw_2p > 180^\circ$ and there is no disk with $v$ and $p$ on its
  boundary that is empty of witnesses.

  As at most $\frac{2}{3}n$ vertices are surrounded by two witnesses,
  $\frac{4}{3}n$ witnesses are sufficient to remove all the edges in
  the Delaunay triangulation of $P$, and the claim follows.
\end{proof}

\subsection{Stabbing squares}\label{subsection:stabbingSquares}

Let $P$ be a set of $n$ points, such that no two of them have
equal abscissa or ordinate, and let $\S$ be the set of
isothetic squares whose boundary contains two points of $P$.
Recall that $\SG^-(P,W)$ is the negative witness square graph of
$P$ with respect to $W$, in which two points $p$ and $q$ from $P$
are adjacent if and only if there is a square that has $p$ and $q$
on its boundary but covers no point from $W$. Equivalently,
$\SG^-(P,W)$ is the Delaunay graph of $P$ with respect to $W$ for
the $L_{\infty}$ metric.

If $p,q\in P$, a set of points $W$ stabs all the
squares whose boundary contains $p$ and $q$ if and only if $pq$ is not
an edge of $\SG^-(P,W)$.  Hence we see that
\[
  st_{\S}(n)=\max_{|P|=n} st_{\S}(P)=\max_{|P|=n} \min
  \{|W| : \SG^-(P,W)=\varnothing\}.
\]
The extrema are taken over pairs of sets $(P,W)$ so that $P
\cup W$ is in general position, i.e., with no two distinct points on
the same vertical or the same horizontal line.

\begin{lemma} \label{lem:disjointSquares}
  There is a set $P$ of $n$ points, no two of them with equal abscissa
  or ordinate, that admits a set of $\frac{5}{4}n-\Theta(\sqrt{n})$
  squares with pairwise disjoint interiors, each with two points of
  $P$ on its boundary.
\end{lemma}

% \muriel[shyly points out, insisting this is not a retailation for
% Ferran's criticism of her figures]{In figure ``Bottom: ...'' in the
%   bottom row, the fourth square from the left has only one point.}
% \boris[comments]{the wording of the last comment is mine, so do not
%   blame Muriel.  She was the one who found this earth-shatterring bug
%   though.}

\begin{figure}[htbp!]
  \centering
  \includegraphics[scale=1]{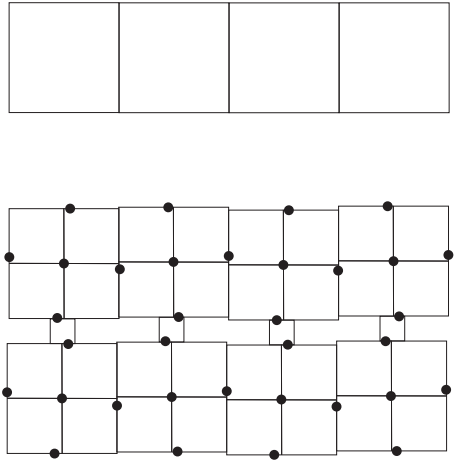}
  \caption{Top: Initial row of basic squares.  Bottom: Two rows of
    basic squares, after perturbation, subdivision, and point insertion,
    are connected by smaller squares.}
  \label{fig:squaresLower}
\end{figure}

\begin{proof}
  Consider a horizontal row of $t$ equal size basic squares each
  sharing vertical sides with its neighbors
  (Figure~\ref{fig:squaresLower}, top).  We apply a different
  infinitesimal vertical translation to each square, and then
  subdivide it into four equal squares; one point is placed at the
  center and four other points very close to the midpoints of the
  initial square edges, as shown in Figure~\ref{fig:squaresLower}.  The
  inserted points are shown in solid, and the union of all of them
  will form the desired set~$P$.

  We place $t$ copies of this construction nearly covertically, but
  applying different slight horizontal shifts to each row, ensuring that no two
  points of $P$ get equal $x$ or $y$ coordinates.  Any two consecutive
  rows are at distance slightly smaller than half the side of the
  original basic square, and we place $t$ connecting squares between
  the two rows, each touching two points of $P$, as in the figure.

  The point set $P$ constructed in this way has a total of $n=4t^2+t$
  points and admits a set of $5t^2-t=\frac54 n-\Theta(\sqrt n)$
  squares with pairwise disjoint interiors, each one with two points
  from $P$ on its boundary.
\end{proof}

\begin{figure}[htbp!]
  \centering
  \includegraphics[scale=1]{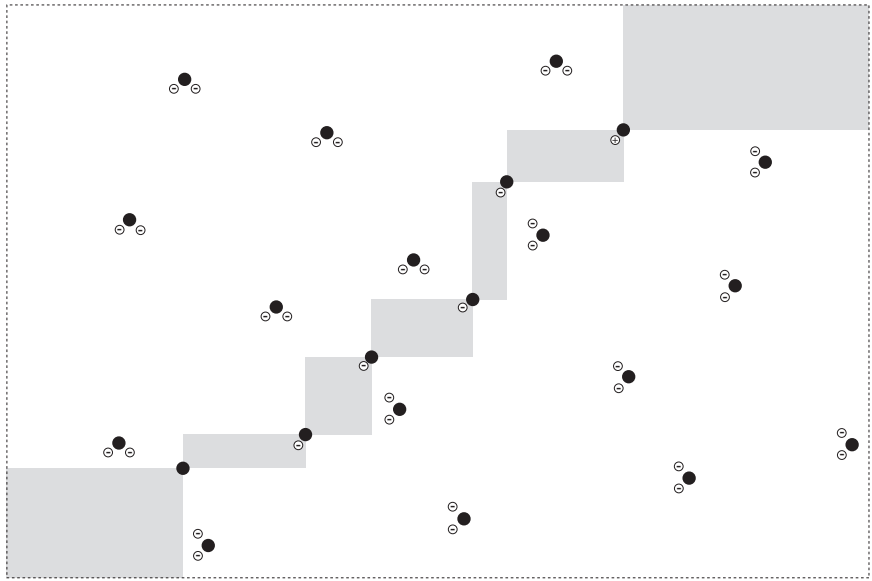}
%  \caption{Constructing for a point set $P$ a set $W$ of positive witnesses such that
%  $\RG^+(P,W)=K_{|P|}.$}
  \caption{Illustration for the proof of
    Theorem~\ref{thm:stabbingSquares}.}
  \label{fig:rectangleSufficient}
\end{figure}

\begin{theorem}
  \label{thm:stabbingSquares}
  The function $st_S(n)$ satisfies $\frac54n-\Theta(\sqrt{n})\le
  st_S(n)\le 2n-\Theta(\sqrt{n})$.
\end{theorem}
\begin{proof}
The lower bound follows from the preceding lemma.  We show that
$2n-\Theta(\sqrt{n})$ witness stabbing points are always
sufficient. Notice that a square containing two points always
contains the rectangle they define as opposite corners: We
prove a stronger claim, namely, that
$2n-\Theta(\sqrt{n})$ points are always sufficient for stabbing
the rectangles such that two opposite corners belong to a given
set $P$ of $n$ points. Using Dillworth's theorem for partially
ordered sets (or Erd\H os-Sz\'{e}keres theorem for sequences) we get a
maximal subset $P'$ of $P$ of at least $\sqrt n$ points with
increasing $x$, such that their ordinates strictly decrease or
strictly increase; we assume the latter without loss of
generality. Consider the boxes that have as opposite corners
consecutive points in this sequence (adding points
$(-\infty,-\infty)$ and $(+\infty,+\infty)$).  The interiors of
these boxes, shown shaded in Figure~\ref{fig:rectangleSufficient},
cannot contain any other point from $P$ because of the maximality
of $P'$.

Let $\varepsilon_x$ and $\varepsilon_y$ be the minimum gap between the
$x$-coordinates and the $y$-coordinates of the points in $P$,
respectively, and define $\varepsilon\mathop{:=}\min \{\varepsilon_x,
\varepsilon_y\}/3$, which is by assumption a positive number.

We put a witness inside every finite shaded box, namely at
position $(x-\varepsilon, y-\varepsilon)$, if $(x,y)$ is the upper
right corner of the box.  For every point $(x,y)\in P\setminus P'$
in the upper bay we put witnesses in its relative third and fourth
quadrant, at positions $(x-\varepsilon, y-\varepsilon)$ and
$(x+\varepsilon, y-\varepsilon)$.  Finally, for every point
$(x,y)\in P\setminus P'$ in the lower bay we put witnesses in its
relative second and third quadrant, at positions $(x-\varepsilon,
y+\varepsilon)$ and $(x-\varepsilon, y-\varepsilon)$ (see Figure~\ref{fig:rectangleSufficient}).
In this way any rectangle with
two opposite corners in $P$ is stabbed, and the total number of
used witnesses is at most $2n-\sqrt n$.
\end{proof}

\begin{figure}[htbp!]
  \centering
  \includegraphics[scale=1]{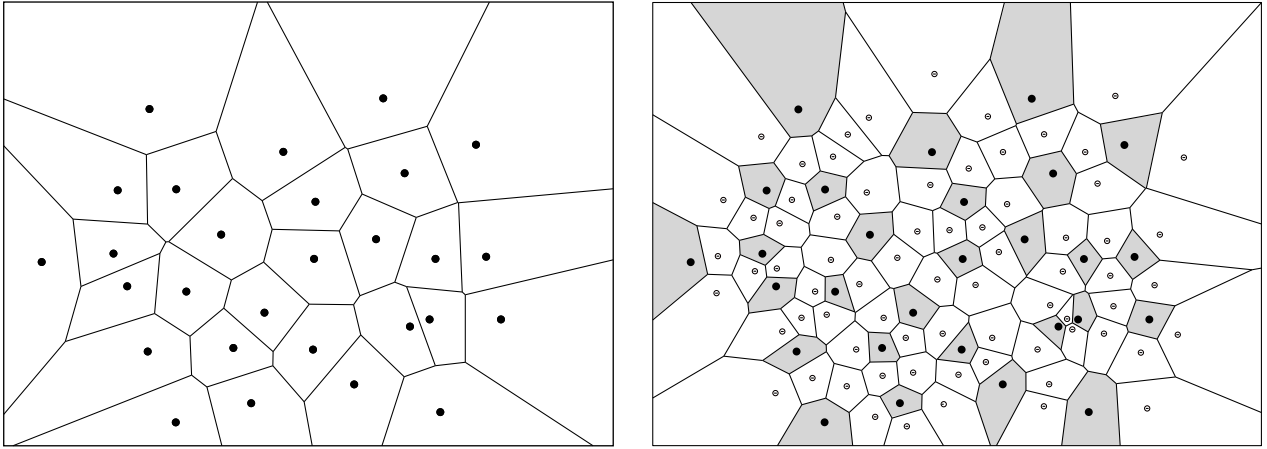}
%  \caption{Constructing for a point set $P$ a set $W$ of positive witnesses such that
%  $\RG^+(P,W)=K_{|P|}.$}
  \caption{The witnesses (right) prevent original points (left) from being Voronoi neighbors with the metric $L_2$.}
  \label{fig:redAndBlueL_2}
\end{figure}

\subsection{Voronoi discrimination}\label{subsection:discrimination}

Given a set $P$ of $n$ ``black'' points in the plane, how many ``white'' points are
needed in the worst case to completely separate the Voronoi regions of the
black points from each other? Observe that this problem is precisely the one we have
been considering throughout this section, as can be formulated in terms of finding a
set $W$ of witnesses (the white points) such that $\DG^-(P,W)=\varnothing$, with the Euclidean metric (Figure~\ref{fig:redAndBlueL_2}),
or that $\SG^-(P,W)=\varnothing$, when the $L_{\infty}$ metric is considered (Figure~\ref{fig:redAndBlueL_infinity})

This interesting discrimination problem seems potentially useful in
several applications.  However, to the best of our knowledge, this
problem had not been explored before, either from the combinatorial
viewpoint, or from the viewpoint of computation.  Various related problems
without satisfactory solutions exist as well, for example, finding
placements for points such that their Voronoi regions will cover
maximal area \cite{DHKS}, delineating boundaries \cite{T03}, or
competing for area as modeled by a two-players game \cite{ACC02}.

\begin{figure}[htbp!]
  \centering
  \includegraphics[scale=1]{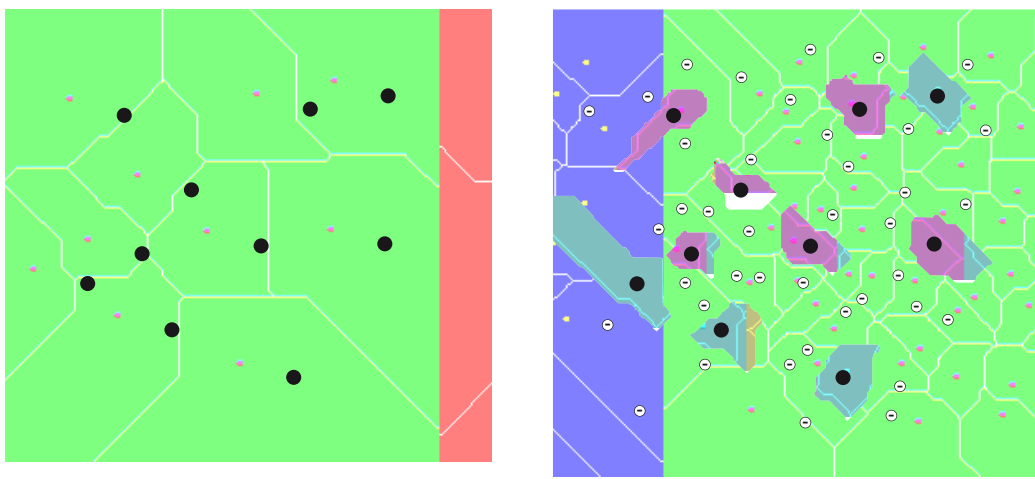}
%  \caption{Constructing for a point set $P$ a set $W$ of positive witnesses such that
%  $\RG^+(P,W)=K_{|P|}.$}
  \caption{The witnesses (right) prevent original points (left) from being Voronoi neighbors with the metric $L_{\infty}$.}
  \label{fig:redAndBlueL_infinity}
\end{figure}

%%%%%%%%%%%%%% JUST COMMENTS %%%%%%%%%%%%%%%%%%%%%%%%%%%%%%%%%%%%%%%%%%%%%%%%
%%% This comments are not included in the only-Delaunay version

%\begin{table}[htbp!]
%\begin{center}
%\begin{tabular}{|l|}
%\hline The number of stabbing witnesses for squares should be smaller than for rectangles.\\
%However a proof would have to include metric arguments beyond combinatorics.  \\
% \hline
%\end{tabular}
%\end{center}
%\end{table}
%%%%%%%%%%%%%%%%%%%%%%%%%%%%%%%%%%%%%%%%%%%%%%%%%%%%%%%%%%%%%%%%%%%%%%%%%%%%%%
%
%\begin{figure}[htbp!]
%  \centering
%  \includegraphics[scale=1]{rectanglesVersusSquares.eps}
%  \caption{This example illustrates the comment above.  A linear number of points is
%  required to stab the rectangles between the left and right subsets, while just two of them kill
%  all the squares.}
%  \label{fig:rectanglesVersusSquares}
%\end{figure}
%

%%%%%%%%%%%%%%%%%%%%%%%%%%%%%%%%%%%%%%%%%%%%%%%%%%%%%%%%%%%%%%%%%%%%%%%%%%%%%%

%%%%%%%%%%%%%%%%%%%%%%%%%%%%%%%%%%%%%%%%%%%%%%%%%%%%%%%%%%%%%%%%%%%%%%%%%%%%%%
\section{Concluding remarks}\label{section:conclusion}

We have introduced in this paper the generic concept of witness graphs and
described several properties and computation algorithms for two specific examples,
one with negative witnesses and Euclidean metric balls as interaction regions,
another one using isothetic squares, the $L_{\infty}$ balls, and positive witnesses.

Several open problems remain. In particular, we have characterized some graphs that
can be realized as witness Delaunay graphs, and some others that cannot. A complete
combinatorial characterization would certainly be desirable.
Closing the gaps between the bounds in Theorem~\ref{thm:stabbingCircles} and
Theorem~\ref{thm:stabbingSquares} on the maximum number of witnesses needed to eliminate
all edges in a witness Delaunay graph, and a square graph, respectively, also seem
to us interesting problems on the combinatorial side.

As for algorithms, it can be easily proved that designing an \emph{output-sensitive}
algorithm for constructing a witness Delaunay graph with $k$ edges has a lower bound complexity
$\Omega(k+n \log n)$, given its set of $n$ vertices and witnesses, yet the most efficient algorithm
we have found has running time, $O(k \log n + n \log^2 n)$,
hence there is still a complexity gap to be resolved.

However, we consider that the most prominent issue in this regard is
that we have not obtained any complexity results on computing an
optimal discriminating set of witnesses for a given point set, i.e.,
given a set $P$, find a minimum set $W$ such that no two Voronoi
regions of points from $P$ are adjacent in the Voronoi diagram
$VD(P\cup W)$, which we know is equivalent to having
$\DG^-(P,W)=\varnothing$, for the Euclidean metric, and
$\SG^-(P,W)=\varnothing$, for $L_{\infty}$.  From practical point
of view, the question of computing efficiently a small (i.e.,
approximating the smallest-size one) discriminating
set seems possibly the most relevant.

\paragraph*{Acknowledgments.} We are grateful to Pankaj K.\
Agarwal for helpful discussions. In particular, all main ideas
underlying the algorithm in Theorem \ref{thm:sweepingGhosts} were suggested by him.

%%% The notes below are not included in the only-Delaunay version
%\noindent \textsc{Internal notes for our eyes}:
%\begin{enumerate}
%\item The optimality of the output-sensitive algorithms has to be
%discussed.
%\item We would like to design an output-sensitive algorithm for
%computing $\DG^{-}(P,W)$.
%\item Add a picture of $P_7$.  Would it pay to investigate which
%  cycles are realizable, and what other small graphs?
%\item When $\RG^{+}(P,W)$ has two non-trivial components the interaction
%between them is very strong and requires further investigation.
%\item Singleton components of $\RG^{+}(P,W)$ have to be investigated.
%\item Examples of graphs that fulfill the conditions of Theorem
%\ref{thm:rectangleStructural} but are not rectangle graphs should
%be constructed.
%\item We would like to improve the bounds in Section~\ref{section:stabbing},
%particularly the upper bound in Theorem~\ref{thm:stabbingSquares}.
%\end{enumerate}
%
%%%%%%%%%%%%%%%%%%%%%%%%%%%%%%%%%%%%%%%%%%%%%%%%%%%%%%%%%%%%%%%%%%%%%%%%%%%%%%

%-----------------------------------------------------------------------------
\end{document}